\documentclass[12pt,conference, onecolumn,compsocconf]{IEEEtran}

\usepackage{hyperref}
\usepackage{graphicx}
\usepackage{amssymb}
\usepackage{amsmath}
\usepackage{dsfont}
\usepackage{url}

\newtheorem{theorem}{Theorem}

\newtheorem{lem}{Lemma}
\newtheorem{prop}{Proposition}

\newcommand{\Reals}{\mathds R}

\newcommand{\bx}{\mathbf{x}}

\newcommand{\bw}{\mathbf{w}}

\newcommand{\bq}{\mathbf{q}}
\newcommand{\bp}{\mathbf{p}}

\renewcommand{\P}{\mathcal{P}}

\newcommand{\eps}{\varepsilon}

\newcommand{\X}{\mathcal{X}}
\newcommand{\C}{\mathcal{C}}

\newcommand{\ra}{\rightarrow}

\renewcommand{\max}{\mathrm{max}}

\long\def\symbolfootnote[#1]#2{\begingroup%
\def\thefootnote{\fnsymbol{footnote}}\footnote[#1]{#2}\endgroup}
\begin{document}

\title{\huge The Dispersion of Lossy Source Coding} 

\newcommand{\Types}{\mathcal{T}}
\newcommand{\ED}{\mathcal{E}\!(D)}
\newcommand{\prob}[1]{\Pr \Bigl\{ #1\Bigr\}}

\author{
\authorblockN{Amir Ingber}
\authorblockA{Dept. of EE-Systems,
TAU\\
Tel Aviv 69978, Israel\\
Email: ingber@eng.tau.ac.il}
\and
\authorblockN{Yuval Kochman}
\authorblockA{EECS Dept.,
MIT\\
Cambridge, MA 02139, USA \\
Email: yuvalko@mit.edu}
}

\maketitle

\begin{abstract}
In this work we investigate the behavior of the minimal rate needed in order to guarantee a given probability that the distortion exceeds a prescribed threshold, at some fixed finite quantization block length. We show that the excess coding rate above the rate-distortion function is inversely proportional (to the first order) to the square root of the block length. We give an explicit expression for the proportion constant, which is given by the inverse $Q$-function of the allowed excess distortion probability, times the square root of a constant, termed the \emph{excess distortion dispersion}.
This result is the dual of a corresponding channel coding result, where the dispersion above is the dual of the channel dispersion.
The work treats discrete memoryless sources, as well as the quadratic-Gaussian case.
\end{abstract}

\section{Introduction}\label{sec:intro}

Rate-distortion theory \cite{Berger_RateDistortionBook} tells us that in the limit of large block-length $n$, a discrete memoryless source (DMS) with distribution $\bp$ can be represented with some average distortion $D$ by a code of any rate greater than the rate-distortion function (RDF)
\begin{align} \label{eqn:RDF}
	R(\bp,D) = \min_{W:E_{\bp,W}[ d(X,\hat X)] \leq D} I(\bp,W),
\end{align}
where $d(x,\hat x)$ is the distortion measure, $W(\hat x|x)$ is any channel from the source to the reproduction alphabet and $I(\cdot,\cdot)$ denotes the mutual information.
However, beyond the \emph{expected} distortion, one may be interested in ensuring that the distortion for one source block is below some threshold. To that end, we define an \emph{excess distortion} event $\ED$ as
\begin{equation}\label{eqn:ED}
    \ED \triangleq \{ d(\bx,\hat\bx)>D \},
\end{equation}
where $d(\bx,\hat\bx) \triangleq \frac{1}{n}\sum_{i=1}^n d(x_i,\hat x_i)$ is the distortion between the source and reproduction words $\bx$ and $\hat \bx$.

A natural question to ask is how fast can the probability of such event be made to decay as a function of the block length. An asymptotic answer is given by Marton's excess distortion exponent  \cite{Marton1974fidelityCriterion}: for the best code of rate $R$,
\begin{align} \label{eqn:Marton_exponent}
\lim_{n\rightarrow\infty} - \frac{1}{n} \log \Pr\{\ED\} = \min_{\bq:R(\bq,D)\geq R} D(\bq\|\bp) \triangleq F(R,\bp,D),
\end{align}
assuming the limit exists. $D(\cdot\|\cdot)$ is the divergence between the two distributions.\footnote{Throughout the paper logarithms are taken with the natural base $e$ and rates are given in nats.} Intuitively speaking, this result means that, asymptotically, the error probability is governed by the first-order empirical statistics of the source sequence; if the sequence happens to be ``too rich'' to be quantized with rate $R$, en error (excess distortion event) will occur.

We are interested in the following related question:  for a given excess distortion probability $\eps$, what is the optimal (minimal) rate required to achieve it? This question is unanswered by Marton's exponent, and even the asymptotical behavior of the optimal rate is unknown.

A similar question can be asked in the context of channel coding: for a given error probability $\eps$, what is the maximal communication rate that can be achieved. Again, this question is unanswered by the channel error exponent \cite{GallagerInfoTheoryBook}. The asymptotics of the rate behavior was first studied in the 1960's \cite{Strassen62_Asymptotische} using the normal approximation. This result was recently tightened and extended to the Gaussian channel, along with nonasymptotic results, in a comprehensive work by Polyanskiy et al. \cite{PolyanskiyPVFiniteLength10}. In channel coding the maximal rate that can be achieved over a channel $W$ is approximately given by
\begin{equation}\label{eqn:channelDispersion}
    R \cong C(W) - \sqrt \frac{V(W)}{n} Q^{-1} (\eps), 
\end{equation}
where $C(W)$ is the channel capacity, $Q$ is the complementary Gaussian cumulative distribution function, and the quantity $V(W)$ is a constant that depends on the channel only, termed the channel dispersion. See \cite{PolyanskiyPVFiniteLength10} for details and more refinements of \eqref{eqn:channelDispersion}.

Our main result is the following. Suppose the source $\bp$ is to be quantized with distortion threshold $D$, and a fixed probability for excess distortion $\eps>0$. Then the minimal rate $R$ needed for quantization in blocks of length $n$ is given by

\begin{align} \label{eqn:R(D)dispersionGeneral}
    R \cong R(\bp,D) + \sqrt\frac{V(\bp,D)}{n} Q^{-1}(\eps),
\end{align}
where $V(\bp,D)$ is a constant which we call the \emph{excess distortion dispersion}, given in detail later on. We show that \eqref{eqn:R(D)dispersionGeneral} holds for any DMS under some smoothness conditions on $R(\bp,D)$, and for a Gaussian source with quadratic distortion measure, see Theorems \ref{thm:discrete} and \ref{thm:QG} respectively.

It is worth noting that that there is a large body of previous work regarding the redundancy of lossy source coding in related setting. However, these works are mostly concerned with two questions: the behavior of the word-length of variable-rate codes where the distortion should always be below some threshold (a.k.a. $D$-semifaithful codes) \cite{YuSpeed93_semifaithful}, or the average excess distortion of fixed-rate codes; see e.g. \cite{ZhangYangWei97},\cite{Kontoyiannis2000} and the references therein. We consider the excess-distortion probability, thus bridging between these works and the concepts of excess-distortion exponent and dispersion discussed above. In this context, the work by Kontoyiannis \cite{Kontoyiannis2000} is of special interest, since it introduces a constant which equals $V(\bp,D)$, see in the sequel.

\section{Main Result for Discrete Memoryless Sources}\label{sec:defsMainResult}

Let the source $X$ be drawn from an i.i.d. distribution $\bp$ over the alphabet $\X = \{1,...,L\}$, and let the reproduction alphabet be $\hat \X = \{1,...,K\}$. The distribution $\bp$ can be seen as a vector $\bp = [p_1,...,p_L]^T \in \P_L$, where $p_i=\Pr(X = i)$ and $\P_L$ is the probability simplex:
\begin{equation}
    \P_L \triangleq \left\{\bq \in \Reals^L |  q_i \geq 0 \forall i\in\{1..L\};\ \sum_{i=1}^Lq_i=1\right\}.
\end{equation}
Let $d:\X\times\hat \X  \ra \Reals^+$ denote a general nonnegative single-letter distortion measure, bounded by some finite $D_\max$. Denote the rate distortion function for the source $\bp$ and the distortion measure $d(\cdot,\cdot)$ at some level $D$ by $R(\bp,D)$. Whenever this function is differentiable w.r.t. its coordinates $p_i$, define the partial derivatives by
\begin{align} \label{eqn:RprimeDef}
	R'(i) \triangleq \left. \frac{\partial}{\partial q_i} R(\bq,D) \right|_{\bq=\bp}.
\end{align}
Note that $R'(i)$ implicitly depends on $\bp$ and $D$ as well. For a random source symbol $X$, we may look at $R'(i)$ as the values that a random variable $R'(X)$ takes.
Also note that in order to define the derivative, we extend the definition of the RDF $R(\bp,D)$ to general vectors in $(0,1)^L$ (cf. \cite[Theorem 2]{YuSpeed93_semifaithful}). In any case, we will only be interested in the value of this derivative for values of $\bp$ within the simplex, i.e. that represent probability distributions.

Let $\bx \in \X^n$ and $\hat\bx \in \hat\X^n$ denote the source and reproduction words respectively.
Recalling \eqref{eqn:ED}, let $R_{\bp,D,\eps}(n)$ be the optimal (minimal) code rate at length $n$ s.t. the probability of an excess distortion event $\ED$ is at most $\eps$.

It is known that $R_{\bp,D,\eps}(n) \ra R(\bp,D)$ as $n\ra\infty$. This can be deduced e.g. by Marton's excess distortion exponent \cite{Marton1974fidelityCriterion}. Our main result quantifies the rate of this convergence.

\begin{theorem} \label{thm:discrete}
A DMS with probability $\bp$ is to be quantized with distortion threshold $D$, block length $n$ and excess distortion probability $\eps$. Assume that $R(\bq,D)$ is differentiable w.r.t. $D$ and twice differentiable w.r.t. $\bq$ in some neighborhood of $(\bp,D)$. Then
\begin{equation}\label{eqn:R(D)boundTheorem1}
    R_{\bp,D,\eps}(n) = R(\bp,D) + \sqrt\frac{V(\bp,D)}{n} Q^{-1}(\eps) + O\left(\frac{\log n}{n}\right),
\end{equation}
where $V(\bp,D)$ is the \emph{excess distortion dispersion}, given by
\begin{align} \label{eqn:R(D)dispersion}
V(\bp,D) \triangleq \mathrm{Var}[R'(X)] = \sum_{i=1}^L p_i (R'(i))^2 - \left[\sum_{x=1}^L p_i R'(i)\right]^2.
\end{align}
\end{theorem}

This result is closely related to the following central-limit theorem (CLT) result of \cite{Kontoyiannis2000}. If we allow a code with variable rate $r(\bx)\triangleq l(\bx)/n$, where $l(\bx)$ is the length of the codeword needed to describe the source word $\bx$, then for the best code:
\begin{equation*}
    r(\bx) = R(\bp,D) + \frac{G_n}{\sqrt{n}} + O\left(\frac{\log n}{n}\right),
\end{equation*}
where $\{G_n\}$ converge in distribution to a Gaussian random variable of variance $V(\bp,D)$.\footnote{The variance has a different expression in \cite{Kontoyiannis2000}, we show in Section \ref{ssec:alternative} that the forms are equivalent.} If $G_n$ are exactly Gaussian, and then we truncate this variable-length code by assuming an excess-distortion event at each time that the length is over $nR$, then the excess distortion probability exactly satisfies the achievability bound of Theorem~\ref{thm:discrete}. However, this is not immediate, as one needs to take into account the rate of convergence of the sequence $\{G_n\}$.

We follow a different direction, which is closer in spirit to the derivation of the excess distortion exponent in \cite{Marton1974fidelityCriterion}. Specifically, we show that the $O(1/\sqrt{n})$ redundancy term comes only from the probability that the source will produce a sequence whose type is too complex to be covered with rate $R$.

The proof is based on the method of types. We adopt the notation of Csisz\'ar and K\"orner \cite{Csiszar81}: The \emph{type} of a sequence $\bx \in \X^n$ is the vector $P_\bx\in\P_L$ whose elements are the relative frequencies of the alphabet letters in $\X$. $\Types_n$ denotes all the types of sequences of length $n$. We say that a sequence $\bx$ has type $\bq \in \Types_n$ if $P_\bx = \bq$. The \emph{type class} of the type $\bq \in \Types_n$, denoted $T_\bq$, is the set of all sequences $\bx \in \X^n$ with type $\bq$.

For a reconstruction word $\hat \bx \in \hat\X$, we say that $\bx$ is $D$-covered by $\hat \bx$ if $d(\bx,\hat \bx) \leq D$.

\begin{prop}[Type covering] \label{prop:types}
Let $\bq\in\Types_n$ with a corresponding type class $T_\bq$. Let $A(\bq,\C,D)$ be the intersection of $T_{\bq}$ with the set of source sequences $\bx \in \X^n$ which are $D$-covered by at least one of the words in a codebook $\C$ with rate $R$ (i.e. $|\C|=e^{nR}$). Then:
\begin{enumerate}
\item If $|\partial R(\bq,D) / \partial D|$ is bounded in some neighborhood of $\bq$, then there exists a codebook $\C_\bq$ that completely $D$-covers $T_\bq$ (i.e. $A(\bq,\C_\bq,D)=T_\bq$), where for large enough $n$,
\begin{align}
    \frac{1}{n}\log|\C_\bq| = R \leq R(\bq,D) + J_1\frac {\log n}{n},
\end{align}
where $J_1=J_1(L,K)$ is a constant.
\item For any type $\bq \in \Types_n$ s.t. $R(\bq,D) > R$, the fraction of the type class that is $D$-covered by any code with rate $R$ is bounded by
    \begin{equation}
        \frac{|A(\bq,\C_n,D)|}{|T_\bq|} \leq \exp\left\{-n\left[R(\bq,D) - R+J_2\frac{\log n}{n}\right] \right\},
    \end{equation}
    where $J_2=J_2(L,K)$ is a constant.
\end{enumerate}
\end{prop}

The first part of this proposition is a refinement of Berger's type-covering lemma \cite{Berger_RateDistortionBook}, found in \cite{YuSpeed93_semifaithful}. The second part is a corollary of \cite[Lemma 3]{ZhangYangWei97}.
Both parts of the proposition are stronger versions than needed in \cite{Marton1974fidelityCriterion}, due to the non-exponential treatment of the excess distortion probability.\footnote{For the first part, Marton uses Berger's original lemma, while for the second part it is proved that the ratio between $|T_\bq|$ and $|A(\bq,\C,D)|$ is upper-bounded by a constant.} Equipped with this, the missing ingredient is an analysis of the relation between the rate $R$ and the probability of the source to produce a type which requires a description rate higher than $R$. It is given in the following lemma which is proved in Section \ref{sec:RateRedundancyLemma}.

\begin{lem}[Rate Redundancy]\label{lem:RateRedundancy} Consider a DMS $\bp$ and a distortion threshold $D$. Assume that $R(\bp,D)$ is differentiable w.r.t. $D$ and twice differentiable w.r.t. $\bp$ at some neighborhood of $(\bp,D)$. A random source word is denoted by $\bx$ and its type by $P_\bx$.
Let $\eps$ be a given probability and let $\Delta R$ be chosen s.t.
\begin{equation*}
    \Pr\{R(P_\bx,D)-R(\bp,D) > \Delta R\} = \eps.
\end{equation*}
Then, as $n$ grows,
\begin{equation}
    \Delta R = \sqrt \frac{V(\bp,D)}{n} Q^{-1}(\eps) + O\left(\frac{\log n}{n}\right),
\end{equation}
where $V(\bp,D)$ is given by \eqref{eqn:R(D)dispersion}.
The same holds even if we replace $\eps$ with $\eps + g_n$, as long $g_n = O\left(\frac{\log n}{\sqrt n}\right)$.
\end{lem}

\begin{proof}[Proof of Theorem \ref{thm:discrete}]
\emph{Achievability part}.

Let $\Delta R>0$. We construct a code $\C$ as follows. The code shall consist of the union of the codes that cover all the types $\bq \in \Phi(n,D,\Delta R)$, where
\begin{equation}\label{eqn:CoveredTypes}
    \Phi(n,D,\Delta R) = \{\bq: R(\bq,D) \leq R(\bp,D) + \Delta R\} \cap \Omega_n,
\end{equation}
where $\Omega_n = \left\{\bq: \|\bp - \bq\|^2 \leq L \frac{\log n}{n} \right\}$.

\begin{lem}\label{lem:annoyingLemma} For a source word $\bx$ drawn from the $\bp$, we have $ \Pr\{P_\bx \notin \Omega_n\} \leq \frac{2L}{n^2}$.
\end{lem}
The proof for this technical lemma is omitted. It can be proved using techniques similar to those in \cite[Theorem 2]{YuSpeed93_semifaithful}.

The size of the code is bounded by
\begin{align}
    |\C| &\leq \sum_{\bq \in \Phi(n,D,\Delta R)} |\C_\bq| \leq |\Types_n| | \C_{\bq^*}| \leq (n+1)^L | \C_{\bq^*}|,
\end{align}
where $\bq^*$ is the largest type class that is covered.

Since we assumed that $R(\bp,D)$ is differentiable w.r.t. $D$ at $\bp$, the derivative is bounded over any small enough neighborhood of $\bp$. In particular, it is bounded over $\Omega_n$ for large enough $n$, thus for all types covered by the codebook. We can thus apply part 1 of Proposition~\ref{prop:types} and we get a bound on the rate:
\begin{align}
    R = \frac{1}{n} \log |\C| \leq & \frac{L}{n} \log (n+1) +  \frac{1}{n}\log|\C_{\bq^*}|\\
    \leq & R(\bp,D) + \Delta R +O\left(\frac{\log n}{n}\right).\label{eqn:rateWithDeltaR}
\end{align}

Since we completely cover all the types in $\Phi(n,D,\Delta R)$, we have that the probability of excess distortion \eqref{eqn:ED} satisfies
\begin{align}
\Pr\{\ED\}=&\prob{P_\bx \notin \Phi(n,D,\Delta R)} \nonumber \\
\leq &  \prob{ R(P_\bx,D) \leq R(\bp,D) + \Delta R} + \Pr\{P_\bx \notin \Omega_n\} \label{eqn:union_in_direct} \\
\leq &  \prob{ R(P_\bx,D) \leq R(\bp,D) + \Delta R} + \frac{2L}{n^2}. \label{eqn:annoying_in_direct}
\end{align}
where \eqref{eqn:union_in_direct} follows from the union bound, and \eqref{eqn:annoying_in_direct} is justified by Lemma \ref{lem:annoyingLemma}.

We select $\Delta R$ s.t. the probability for $\{R(P_\bx,D)>R(\bp,D)+\Delta R\}$ is exactly $\eps -  \frac{2L}{n^2}$, and get a code with excess distortion probability at most $\eps$. By Lemma~\ref{lem:RateRedundancy} we have
\begin{equation*}
    \Delta R = \sqrt \frac{V(\bp,D)}{n} Q^{-1}(\eps) + O\left(\frac{\log n}{n}\right),
\end{equation*}
and by plugging into \eqref{eqn:rateWithDeltaR} the rate $R$ is bounded by the RHS of \eqref{eqn:R(D)boundTheorem1}, as required.

\emph{Converse part}.

Let $\C$ be a code with rate $R$, and suppose that its excess distortion probability is $\eps$. Our goal is to lower bound $\Delta R = R - R(\bp,D)$.

Again, the source word is $\bx$ and its type is $P_\bx$. The following holds for any $\Psi$:
\begin{align}
    \eps = \Pr\{\ED\} =& \prob{\ED|R(P_\bx,D) \leq R+\Psi} \prob{R(P_\bx,D) \leq R+\Psi} \nonumber\\
    &+\prob{\ED|R(P_\bx,D) > R+\Psi} \prob{R(P_\bx,D) > R+\Psi}\nonumber\\
    \geq& \prob{\ED|R(P_\bx,D) > R+\Psi} \prob{R(P_\bx,D) > R+\Psi}.\label{eqn:eps_LB}
\end{align}

Take a type $\bq \in \Types_n$, and assume that $R(\bq,D) > R + \Psi$.
By the second part of Proposition~\ref{prop:types}, the fraction of the type class $T_\bq$ that is covered by the code $\C$ is at most
\begin{align}
    \exp\left\{-n\left[R(\bq,D) - R+J_2\frac{\log n}{n}\right] \right\}\leq \exp\left\{-n\Psi +J_2\log n \right\}
\end{align}
By setting $\Psi = (J_2 + 1)\frac{\log n}{n}$ we get that the fraction is bounded by $1/n$. Since the source sequences within a given type are uniformly distributed, we get that the probability of covering a sequence from a type that its $R(P_\bx,D)$ is too high is at most $1/n$. We therefore have
\begin{align}
    \eps &\geq \left(1-\frac{1}{n}\right)\prob{R(T_\bx,D) > R+\Psi}\nonumber\\
    &\geq\frac{1}{1+\frac{2}{n}}\prob{R(T_\bx,D) > R+\Psi},\label{eqn:Converse_eps_LB1}
\end{align}
where the last inequality follows since $1-x \geq \frac{1}{1-2x}$ for all $x \in [0,1/2]$.


We rewrite \eqref{eqn:Converse_eps_LB1} and get that $\Delta R$ must satisfy
\begin{equation}
    \eps\left(1+\frac{2}{n}\right) \geq \prob{R(T_\bx,D)-R(\bp,D) > \Delta R +\Psi}.
\end{equation}
By Lemma \ref{lem:RateRedundancy} and the fact that $\Psi = O\left(\frac{\log n}{n}\right)$, we get
\begin{equation}
    \Delta R \geq \sqrt \frac{V(\bp,D)}{n} Q^{-1}(\eps) + O\left(\frac{\log n}{n}\right),
\end{equation}
as required.
\end{proof}
\section{Excess Distortion Dispersion: Properties and Evaluation}\label{sec:properties}

\subsection{Differentiability of the RDF}

In the results above, we assumed differentiability of the RDF $R(\bp,D)$ with respect to $D$ (once) and $\bp$ (twice). In general, the RDF is not differentiable w.r.t. either. However, it is differentiable ``almost always'' in the following sense. Let $K'(\bp,D)$ be the ``effective reproduction alphabet size'', i.e., the number of reproduction letters of positive probability for the channel minimizing \eqref{eqn:RDF}. Then, if $K'(\bp,D)$ is constant in a neighborhood of $D$, then $R(\bp,D)$ is differentiable w.r.t. $D$ and twice differentiable w.r.t. $\bp$ at that point.

When keeping $\bp$ fixed and changing $D$, such points may represent ``jumps'' in the excess distortion dispersion $V(\bp,D)$. In these points, we can not specify the exact behavior of the excess rate, but careful derivation should verify that it is between $V(\bp,D^-)$ and $V(\bp,D^+)$. However, in the process we will encounter at most $L-2$ such points.

\subsection{Alternative Representations}\label{ssec:alternative}

The evaluation of the the excess distortion dispersion seems to be a difficult task, as it involves derivatives of the RDF w.r.t. the source distribution. However, we have the following alternative representations.

First we connect the dispersion to the excess-distortion exponent \eqref{eqn:Marton_exponent}, much in the same way that the channel dispersion constant is related to the channel error exponent; See \cite{PolyanskiyPVFiniteLength10} for details on the early origins of this approximation by Shannon.

\begin{prop} \label{prop:dispersion-exponent}
If $R(p,D)$ is differentiable at distortion level $D$, then \[ V(p,D) = \left[ \left. \frac{\partial^2}{\partial R^2} F(R,p,D) \right|_{R=R(p,D)} \right]^{-1}.\]
\end{prop}

The proof, not included in this version, follows by directly considering the exponent definition \eqref{eqn:Marton_exponent} in the limit of small excess rate.

We further show equivalence to the variance of the excess rate in \cite{Kontoyiannis2000}, which is close in spirit to the dispersion as discussed in Section \ref{sec:defsMainResult}:
\begin{prop} \label{prop:Kont}
If $R(\bp,D)$ is differentiable at distortion level $D$, then $V(p,D)=\mathrm{Var}[f(X)]$ where
\[ f(i) = -\log E_{\hat X} \exp\{-\lambda[d(x_i,\hat x)-D]\}, \]
where the expectation is taken according to the reproduction distribution induced by the channel minimizing \eqref{eqn:RDF} for $\bp$ and $D$, and $\lambda = \partial R(\bp,D) / \partial D$ at that point. \end{prop}

This form is especially appealing, since it can also be shown that $R(\bp,D)=E\{f(X)\}$, thus presenting the dispersion as a ``second-order RDF''. The equivalence can be proven by starting from the RDF presentation above. Applying \eqref{eqn:R(D)dispersion},
\[ V(\bp,D) =  \mathrm{Var}  \left\{ \left. \frac{\partial}{\partial q_i} \sum_{j=1}^L q_j f(j) \right|_{\bq=\bp} \right\} =
\mathrm{Var}  \left\{ f(i) + \sum_{j=1}^L p_j \cdot \left. \frac{\partial f(j)}{\partial q_i}\right|_{\bq=\bp} \right\} . \] Straightforward derivation shows that the term to the right of the addition in the last form is constant in $i$, thus it does not effect the variance, as required.

\subsection{Some Special Cases}

In some cases the evaluation may be simplified, as follows.

\begin{enumerate}
\item \textbf{Zero distortion}. Whenever $R(\bp,0) = H(\bp)$, we have \[ R'(i) = \left. \frac{\partial}{\partial q_i} H(q)\right|_{q=p} = -1 - \log p_i. \]
Thus,
\begin{align} \label{dispersion_constant_simple}
	V(\bp,0) = \mathrm{Var} \{ \log p_i \}.
\end{align}
This is in agreement with the long known lossless dispersion result \cite{Strassen62_Asymptotische}.

\item \textbf{Difference distortion measure with low distortion}. Assume that \[d(x,\hat x) = d([x-\hat x] \bmod L) \triangleq d(z). \] Since we assumed that each source letter has positive probability, there exists some $D_0(\bp)>0$ s.t. for all $D\leq D_0$ the optimum backword channel is $x = \hat x + z$. The RDF is then given by
\begin{align} \label{difference_RDF}
R(\bp,D) = H(\bp) - H(\bw_z)
\end{align}
where $\bw_z$ is the maximum-entropy distribution such that $E\{d(z)\} \leq D$ \cite[Sec. 4.3.1]{Berger_RateDistortionBook}. Since this distribution is $D$-independent as long as $D < D_0(\bp)$, we have that the second term in \eqref{difference_RDF} is fixed in $\bp$ in a neighborhood of the source distribution. Consequently the derivatives only come from the first term, and \eqref{dispersion_constant_simple} holds for all $0\leq D < D_0$.

\item \textbf{Hamming distortion measure}. In this special case of a difference distortion measure, the optimum backward channel is modulo-additive also above $D_0$, where the modulo is taken over a reduced alphabet. Consequently, the dispersion is the variance of the logarithm of a normalized smaller-alphabet distribution.

\item \textbf{Zero dispersion}. The dispersion becomes zero when the source distribution maximizes the RDF over all possible source distributions among the input alphabet (thus the rate redundancy in Lemma~\ref{lem:RateRedundancy} is zero). Note that this is in agreement with the fact that for this case the excess-distortion exponent ``jumps'' from zero to infinity at zero excess rate. For difference measures, this happens if and only if the source is uniform, in agreement with the observation in \cite{Kontoyiannis2000}. However, in general $\bp$ need not be uniform.

\end{enumerate}

\section{Gaussian Source with Quadratic Distortion Measure}\label{sec:QG}

In this section we part with the assumption that the source is discrete. While the derivation of the excess distortion dispersion for general continuous-amplitude sources is left for future work, we solve the important special case of Gaussian source with MSE (quadratic) distortion measure.

Let the source $X$ be i.i.d. zero-mean Gaussian with variance $\sigma^2$. The distortion measure is given by: $d(x,y) = (x-y)^2$.
For $D\leq\sigma^2$, the quadratic-Gaussian RDF is given by:
\begin{align} \label{eqn:Gaussian_RDF}
R(\sigma^2,D) = \frac{1}{2} \log \left(\frac{\sigma^2}{D}\right).
\end{align}
In this case, the excess distortion exponent \eqref{eqn:Marton_exponent} is given by \cite{IharaKubo2000}:
\begin{align} \label{eqn:Gaussian_exponent}
F(R,\sigma^2,D) = \frac{1}{2}\left[\frac{D}{\sigma^2}e^{2R}-1-\log\left(\frac{D}{\sigma^2}e^{2R}\right)\right] = \frac{e^{2\Delta R}-1-2\Delta R}{2},
\end{align}
where $\Delta R = R - R(\sigma^2,D)$.

As in the finite alphabet case, we define $R_{\sigma^2,D,\eps}(n)$ to be the minimal code rate at length $n$ s.t. the excess distortion probability is at most $\eps$. From the excess distortion exponent \eqref{eqn:Gaussian_exponent} it follows that $R_{\sigma^2,D,\eps}(n) \ra R(\sigma^2,D)$ as $n\ra\infty$.

We are interested in the behavior of $R_{\sigma^2,D,\eps}(n)$ as $n$ grows. We show that the quadratic-Gaussian case behaves according to \eqref{eqn:R(D)dispersionGeneral} just like the finite-alphabet one. Recalling Proposition~\ref{prop:dispersion-exponent}, one expects the dispersion constant to be
\begin{equation*}
     V(\sigma^2,D) = \left[\left. \frac{\partial^2}{\partial R^2} F(R,\sigma^2,D) \right|_{R=R(\sigma^2,D)}\right]^{-1} = \frac{1}{2}.
\end{equation*}
It can also be shown that the value of $\frac{1}{2}$ can be obtained by a continuous version of \eqref{eqn:R(D)dispersion}.

We now show that this is the case indeed.

\begin{theorem}\label{thm:QG}
Let $\eps>0$ be a given excess distortion probability. Then the rate $R_{\sigma^2,D,\eps}(n)$ satisfies
\begin{equation}\label{eqn:QGdispersion}
    O\left(\frac{1}{n}\right) \leq R - R(\sigma^2,D) - \sqrt\frac{1}{2n} Q^{-1}(\eps) \leq \frac{5}{2n}\log n + O\left(\frac{1}{n}\right)
\end{equation}

\begin{proof}[Proof outline]
The proof is similar in spirit to the proof of Theorem \ref{thm:discrete}, where spheres take the part of types. The type class of types near the source distribution is analogous to a sphere with radius $r$, where $r^2$ is close to $n\sigma^2$.

For the achievability part, we define a ``typical'' sphere with radius $\sqrt{n\sigma^2(1+\alpha_n)}$ with $\alpha_n \ra 0$ as $n\ra\infty$. $\alpha_n$ is chosen s.t. the probability that the source falls outside the sphere is exactly $\eps$, so our code needs to $D$-cover the entire sphere. Note that the radius is just over the typical radius of the source. We use a sphere covering result by Rogers \cite[Theorem 3]{RogersSpheres1963}, and find a code that can $D$-cover the entire typical sphere with no more than $c n^{5/2} \left({\sigma^2(1+\alpha_n)/D}\right)^{n/2}$ reconstruction words for some constant $c$. By arguments similar to those used in the proof of Lemma \ref{lem:RateRedundancy} we get $\alpha_n = \sqrt{2/n} Q^{-1}\left(\eps \right) + O\left(\frac{1}{n}\right)$,
so the rate $R$ is bounded according to \eqref{eqn:QGdispersion}.

For the converse part, we follow the proof of the converse to the excess distortion exponent in \cite{IharaKubo2000}. We get that the excess distortion probability is lower bounded by the probability to leave a sphere that has a volume of $e^{nR}$ times the volume of a single $D$-ball around a reconstruction point. Again, using the Berry-Esseen theorem we connect excess distortion probability and the ratio of the radiuses, and get that the rate $R$ is lower bounded according to \eqref{eqn:QGdispersion}.

\end{proof}
\end{theorem}

\section{Proof of the Rate Redundancy Lemma}\label{sec:RateRedundancyLemma}

\begin{proof}[Proof of Lemma \ref{lem:RateRedundancy}]
Let $\bx$ be a source word with type $P_\bx$, drawn from the source $\bp$. We prove the more general version of the lemma, with $\eps + g_n$ being the given probability.
The relation between $\eps$ and $\Delta R$ is given by
\begin{equation}\label{eqn:eps-DeltaR}
    \eps + g_n = \Pr\left\{R(P_\bx,D) > R_\bp(D) + \Delta R\right\}.
\end{equation}

By the regularity assumptions on $R(\bp,D)$, we use the Taylor approximation and write
\begin{equation}
    R(P_\bx,D) = R(\bp,D) + \sum_{i=1}^L (P_\bx(i) - p_i)R'(i) + \gamma(P_\bx,\bp),
\end{equation}
where $R'(\cdot)$ was defined in \eqref{eqn:RprimeDef}, and $\gamma(P_\bx,\bp)$ is the correction term for the approximation. Equation \eqref{eqn:eps-DeltaR} now becomes
\begin{equation}
    \eps + g_n = \Pr\left\{\sum_{i=1}^L (P_\bx(i) - p_i)R'(i) + \gamma(P_\bx,\bp) > \Delta R\right\}.
\end{equation}

By the Taylor approximation theorem, and by the assumption of finite second derivatives of $R(\bp,D)$, we have that the correction term $\gamma(P_\bx,\bp) = O(\|P_\bx-\bp\|^2)$. This means that there exists a constant $\eta$, s.t. for large enough $n$, $\gamma(P_\bx,\bp) < \eta\|P_\bx-\bp\|^2$. By Lemma \ref{lem:annoyingLemma} there exists $\Gamma = O(\log n / n)$ s.t. $\Pr\{\gamma(P_\bx,\bp) > \Gamma \} \leq \frac{2L}{n^2}$.

Using simple probability rules, for any random variables $A$ and $B$ and a constant $c$, we have that for any $\Gamma_1,\Gamma_2$ the following holds:
\begin{align}
    \prob{A+B > c} &\leq \prob{A > c-\Gamma_1} + \prob{B > \Gamma_1}\label{eqn:ElementaryUpper}\\
    \prob{A+B > c} &\geq \prob{A > c+\Gamma_2} - \prob{B < -\Gamma_2}\label{eqn:ElementaryLower}
\end{align}

In our case, we use \eqref{eqn:ElementaryUpper} (resp. \eqref{eqn:ElementaryLower}) to show the upper (resp. lower) bound on $\Delta R$.
By selecting $\Gamma_1 = \Gamma_2 = \Gamma$, we get
\begin{align}
    \eps + g_n &\leq \Pr\left\{\sum_{i=1}^L (P_\bx(i) - p_i)R'(i) > \Delta R - \Gamma \right\} +  O\left(\frac{1}{n^2}\right),\label{eqn:BeforeBerryEssen1}\\
    \eps + g_n &\geq \Pr\left\{\sum_{i=1}^L (P_\bx(i) - p_i)R'(i) > \Delta R + \Gamma \right\} -  O\left(\frac{1}{n^2}\right).\label{eqn:BeforeBerryEssen2}
\end{align}
Now consider the probability expression in \eqref{eqn:BeforeBerryEssen1}:
\begin{align*}
    \Pr\left\{\sum_{i=1}^L (P_{\bx}(i) - p_i)R'(i)>\Delta R-\Gamma\right\}  = \Pr\left\{\frac{1}{n}\sum_{k=1}^n R'(x_k) -  \sum_{i=1}^L p_iR'(i)>\Delta R-\Gamma\right\}.
\end{align*}
$\frac{1}{n}\sum_{k=1}^n R'(x_k)$ can be interpreted as an average of $n$ i.i.d. random variables $R'(X)$, whose expectation is given by $E[R'(X)]=\sum_{i=1}^L p_iR'(i)$. Their variance is given by $V(\bp,D)$, defined in \eqref{eqn:R(D)dispersion}. By the central limit theorem, the sum of i.i.d. random variables normalized by $\sqrt n$ converges to a Gaussian random variable as $n$ grows. Specifically, by the Berry-Esseen theorem (see, e.g. \cite[Ch. XVI.5]{Feller_1971}), we get
\begin{align}
    &\Pr\left\{\frac{1}{\sqrt n}\sum_{k=1}^n \left( R'(x_k) -  E[R'(X)]\right) >\sqrt n (\Delta R-\Gamma)\right\}\nonumber\\
    &= Q\left((\Delta R-\Gamma)\sqrt \frac{n}{V(\bp,D)}\right) \pm \frac{6\xi}{\sqrt n},
\end{align}
where $\xi = E\left[|R'(X)-E[R'(X)]|^3\right]$. By applying the same derivation $\Delta R + \Gamma$, \eqref{eqn:BeforeBerryEssen1} and \eqref{eqn:BeforeBerryEssen2} can be written together as
\begin{equation}
    \eps + O\left(\frac{\log n}{\sqrt n}\right) = Q\left((\Delta R \pm\Gamma )\sqrt\frac{n}{V(\bp,D)} \right).
\end{equation}
By the smoothness of $Q^{-1}(\cdot)$ around $\eps$ and the Taylor approximation we have
\begin{align}
    \Delta R = \sqrt \frac{V(\bp,D)}{n} Q^{-1}\left(\eps\right) + O\left(\frac{\log n}{n}\right),
\end{align}
as required.
\end{proof}

\let\oldbibliography\thebibliography
\renewcommand{\thebibliography}[1]{%
  \oldbibliography{#1}%
  \setlength{\itemsep}{.35em}%
}

\bibliographystyle{IEEEtran}
\bibliography{Master}

\end{document}